\documentclass[11pt]{amsart}

\usepackage{calc}
\usepackage{graphicx}
\usepackage{color}
\usepackage{dcolumn}
\usepackage{bm}
\usepackage[utf8]{inputenc}
\usepackage[T1]{fontenc}
\usepackage{mathptmx}
\usepackage{amsthm}
\usepackage{amsmath}
\usepackage{amsxtra}
\usepackage{amssymb}
\usepackage{natbib}

\newtheorem{theorem}{Theorem}
\newtheorem{proposition}{Proposition}
\newtheorem{example}{Example}

\newcommand{\acr}{a}
\newcommand{\rev}{\vec r}

\newcommand{\hnls}{\epsilon}

\newcommand{\wavefn}{\psi}
\newcommand{\boi}[2]{{]#1,#2[}}

\newcommand{\dnconv}{\searrow}

\newcommand{\otensor}{\otimes}
\newcommand{\ncurl}{\nabla\crossp}
\newcommand{\crossp}{\times}
\newcommand{\ndiv}{\nabla\dotp}
\newcommand{\const}{\text{const}}
\newcommand{\Lap}{\Delta}
\newcommand{\eqv}{\Leftrightarrow}

\newcommand{\csep}{,\quad}
\newcommand{\dotp}{\cdot}

\newcommand{\qeq}{\quad\eqv\quad}
\def\eqdef{=:}
\newcommand{\half}{\frac12}
\newcommand{\subeq}[2]{\mathord{\underbrace{\mathop{#1}}_{#2}}}

\newcommand{\defm}[1]{\emph{#1}}

\newcommand{\vv}{\vec v}

\newcommand{\drat}{r}
\newcommand{\pV}{\beta}
\newcommand{\xx}{\vec x}
\newcommand{\vort}{\omega}
\newcommand{\phs}{\sigma}
\newcommand{\amp}{a}
\newcommand{\Vpot}{V}
\newcommand{\vpot}{\varphi}
\newcommand{\gisen}{\gamma}
\newcommand{\pp}{p}
\newcommand{\ppf}{\hat p}
\newcommand{\dens}{\varrho}
\newcommand{\idens}{v}
\newcommand{\Mach}{M}
\newcommand{\csnd}{c}
\newcommand{\hpm}{h}
\newcommand{\hpmf}{\hat h}
\newcommand{\Hpm}{H}

\newcommand{\epm}{e}
\newcommand{\spm}{s}
\newcommand{\Temp}{T}
\newcommand{\fod}{f}
\renewcommand{\vec}{\mathbf}
\newcommand{\ve}{u}
\renewcommand{\vv}{\vec\ve}
\newcommand{\vy}{\ve^y}
\newcommand{\vx}{\ve^x}
\newcommand{\dat}[3]{\Big(\frac{\partial#1}{\partial#2}\Big)_{#3}}
\newcommand{\dati}[3]{(\partial#1/\partial#2)_{#3}}
\newcommand{\vn}{\vec n}
\newcommand{\vt}{\vec t}

\newcommand{\zz}{\vec z}
\newcommand{\bM}{M}
\newcommand{\bE}{E}
\newcommand{\bv}{\vec v}
\newcommand{\bP}{\vec p}

\setcitestyle{authoryear,open={[},close={]}}

\begin{document}

\title[Triple points]{Triple points and sign of circulation}

\author{Volker Elling}\email{velling@math.sinica.edu.tw}\address{Department of Mathematics, Academia Sinica, Taipei}

\begin{abstract}
  Mach reflection generally produces a contact discontinuity whose circulation has previously been analyzed 
  using ``thermodynamic'' arguments based on the Hugoniot relations across the shocks.
  We focus on ``kinematic'' techniques that avoid assumptions about the equation of state, using only jump relations for conservation of mass and momentum, but not energy. 
  We give a new short proof for non-existence of pure (no contact) triple shocks, recovering a result of Serre. 
  For MR with a zero-circulation but nonzero-density-jump contact we show that the incident shock must be \emph{normal}.
  Nonexistence without contacts generalizes to two or more incident shocks if we assume all shocks are compressive. 

  The \emph{sign} of circulation across the contact has previously been controlled with entropy arguments, showing the post-Mach-stem velocity is generally smaller. 
  We give a kinematic proof assuming compressive shocks and another condition, for example backward incident shocks, or a weak form of the Lax condition.
  
  We also show that for 2+2 interactions (two ``upper'' shocks with clockwise flow meeting two ``lower'' shocks with counterclockwise flow in a single point)
  the circulation sign can generally not be controlled. For $\gamma$-law pressure 
  we show 2+2 interactions without contact must be either symmetric or antisymmetric, with symmetry favored at low Mach number and low shock strength.
  
  For full potential flow instead of the Euler equations we surprisingly find, contrary to folklore and prior results for other models,
  that pure triple shocks without contacts are possible, even for $\gisen$-law pressure with $1<\gisen<3$.
\end{abstract}

\maketitle

\section{Introduction}

\subsection{Background}

\begin{figure*}
  \hfil\input{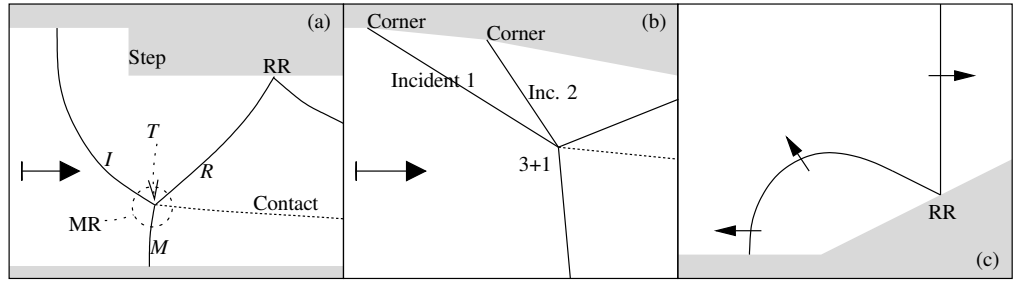}\hfil
  \caption{(a) Steady bow shock ahead a forward-facing step in supersonic flow, Mach reflection (MR) followed by regular reflection (RR).
    (b) Steady 3+1 MR from a two-corner ramp. (c) Pseudo-steady RR at a ramp.}
  \label{fig:mrrr}
\end{figure*}

\begin{figure*}
  \hfil\input{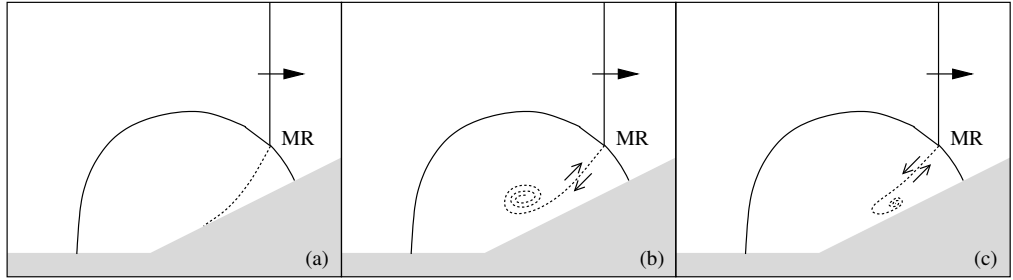}\hfil
  \caption{(a) Pseudo-steady MR with slip line ending at the wall. (b) Slip line ends in a clockwise spiral.
    (c) Slip line reversing and forming a narrow jet near the wall with counterclockwise spiral end.}
  \label{fig:mrsheets}
\end{figure*}

It is well-known that three shock waves cannot meet in a triple point alone without other waves, for example a contact discontinuity. 
These Mach reflections (MR) generally occur when shocks interact with each other or with a solid wall (fig.\ \ref{fig:mrrr} and \ref{fig:mrsheets}

The ``standard'' sign of velocity jump across the contact is so that in the triple point frame the speed is lower on the Mach stem side of the contact. This was already observed by \citet[p.\ 264]{neumann-1943}:
\begin{quote}
``But this compression occurs in the `upper' half in two stages $(I,R)$ and in the `lower' half in one $(M)$. By all experience, theoretical as well as experimental, the former process is less irreversible than the latter --- essentially because it is less abrupt. Hence the substance which crossed `above' $T$ may be expected to have a lower entropy than that which crossed `below' $T$. [...] Hence we may expect, in general, that the `upper' flow has also lower temperature, lower density, and (in this frame of reference!) higher kinetic energy, i.e.\ velocity.'' 
\end{quote}
Other paragraphs in a similar style have sparked debate, 
for example criteria for transition between regular and Mach reflection
\citep{henderson-lozzi,ben-dor-book,li-bendor-parametric,ben-dor-shockwaves2006,elling-sonic-potf}
or the discussion of weak vs.\ strong reflected shock 
\citep{teshukov,hornung-weakstrong,elling-liu-rims05,elling-liu-pmeyer}.
But the quoted argument can be made rigorous easily at least in the case of polytropic equation of state $\pp\idens=(\gisen-1)\epm$, 
where $\gisen$ is ratio of specific heats, $\pp$ pressure, $\idens$ specific volume, $\epm$ energy per mass: 
taking the logarithm of von Neumann's equation (19), which is essentially a Hugoniot relation, shows the jump $[\log\idens]$ across a shock is a \emph{strictly convex} function of $[\log\pp]$. 
Since $[\log\pp]$ is zero across the contact it is additive for (compressive) shocks, so $\log\idens$ increase across the Mach stem $M$ is larger than across the upper shocks $R,I$, 
which implies corresponding results for entropy per mass $\spm$, speed $|\vv|$ etc.

Another discussion of pure (contact-free) triple shocks for $\gisen$-law fluid can be found in \citet[par.\ 129]{courant-friedrichs}.

The convexity argument above does not extend to more general equations of state without effort.
\cite{henderson-menikoff} give an entropy-based proof, under some thermodynamic assumptions 
including restrictions on the Gr\"uneisen coefficient $\frac{\idens}{\Temp}(\frac{\partial\pp}{\partial\spm})_{\idens}$, 
where $\Temp$ is temperature while subscript $\idens$ indicates a partial derivative with respect to entropy per mass $\spm$ is taken along a curve of fixed $\idens$. 

The entropy-based proofs described above are rather ``thermodynamic'': in the conservation of energy relation $[\hpm+|\vv|^2/2]=0$ ($\hpm=\epm+\pp\idens$ enthalpy per mass)
the velocity $\vv$ is eliminated using conservation of mass $[\dens\vv\dotp\vn]=0$ and momentum $[\dens\vv\vv\dotp\vn+\pp\vn]=0$. 
The resulting Hugoniot relation is a constraint on pre- and post-shock $\idens,\spm$ derived purely from the equation of state $\epm=\epm(\idens,\spm)$.
In contrast \citet[Theorem 2.3]{serre-hbfluidmech} gives an elegant \emph{kinematic} proof of nonexistence of pure (contact-free) triple shocks.
``Kinematic'' means the argument avoids the equation of state by not using conservation of energy, and by using in conservation of momentum a pressure $\pp$ that is arbitrary
rather than constrained by $\pp=-\dati{\epm}{\idens}{\spm}$.
Seeing that this is sufficient to disprove \emph{zero} circulation, naturally one wonders whether it is enough to control the \emph{sign} as well.

The sign of circulation of the contact is important for various purposes;
e.g.\ in pseudo-steady Mach reflection at a ramp (fig.\ \ref{fig:mrsheets}) 
the contact cannot attach to the wall if circulation is ``non-standard'' unless the velocity field near the attachment point is rather strongly expanding in the wall-tangential direction \citep{elling-vortexcusps}.
A narrow jet may form instead (fig.\ \ref{fig:mrsheets}(c); \citep{henderson-vasilev-bendor-elperin-2003}). 
The rollup direction of the spiral end is determined by the circulation sign. 
The sign of vorticity determines whether it induces in near-wall velocity a change in downstream or upstream direction, the latter possibly causing near-stagnation and boundary layer separation.

``Thermodynamic'' arguments become very difficult once the equation of state becomes complex. Examples:\\
$\bullet$ Phase transitions like evaporation \citep{zamfirescu-guardone-colonna-jfm2008,nannan-guardone-colonna-pof2014}\\
$\bullet$ Dissociation \citep{grasso-paoli-pof2000}, especially molecular vs.\ atomic oxygen/nitrogen in hypersonic flight \citep{grover-torres-schwartzentruber-pof2019}, classically at bow shocks ahead of spacecraft during atmospheric reentry \citep{boyd-candler-levin-pof1995}\\
$\bullet$ More generally chemical equilibrium reactions, whose chemical species and thermodynamic potentials can be refined to arbitrary complexity \\
$\bullet$ Irreversible transitions from metastable states, especially deonation waves in explosives, where Mach reflections can be observed \citep{bdzil-short-jfm2017,short-quirk-annrevfluidmech2018,yakunliu-jianpingyin-zhijunwang-pof2019}, or burning fuel-air mixtures in scramjets or similar engines
where transition between regular vs.\ Mach or other irregular reflections is closely related to engine unstart (\citep{nanli-juntaochang-etal-pof2018}, also \cite{li-bendor-parametric}).

Shocks in astrophysics should be mentioned as well, although in many cases low densities, plasma effects and radiation may not allow assuming thermal equilibrium.

Even for simpler fluids, seemingly natural conditions such as positive fundamental derivative (genuine nonlinearity) 
may not be universal \citep{bethe,zeldovich-raizer,lambrakis-thompson-1972,colonna-guardone-nannan-2007}.
Assumptions on the equation of state are often derived from stability considerations for thermodynamic states, shocks or viscous shock layers (\citep{randolph-smith,fowles-jfm1981}, also \cite[Section VI]{menikoff-plohr}). 
But shear contacts in particular are rarely stable anyway, yet when breakdown produces complex layers that resemble a contact/shock at larger scale, 
aggregate jump relations must still be satisfied across the layer.

``Kinematic'' proofs are attractive because they avoid equation of state modelling, allowing results with shorter proofs that apply to a much wider range of practical flows.

\subsection{Overview}

In section \ref{section:zerocirc} we give a new kinematic proof of pure triple shock nonexistence (theorem \ref{th:twoplusonetriple}). 
The new proof technique, a reduction to two-particle collision, allows generalizations to multiple particles representing multiple incident shocks (fig.\ \ref{fig:kshocks}, theorem \ref{th:kplusonedensity}).
If we weaken assumptions to allow a contact with nonzero density jump but still zero velocity jump 
we find the incident shock must be normal and the density jump positive (theorem \ref{th:twoplusonedensity}). 

In section \ref{section:circulation} we consider the sign of circulation on the contact; example \ref{ex:negcircposrjump} shows compressiveness of shocks is not sufficient to control it, 
so we add additional assumptions: Proposition \ref{prop:backward-p} shows circulation is ``standard'' if the velocities between the upper shocks are negative.
Theorem \ref{th:circ-backward} shows this is true for incident shocks that are \emph{backward} (outward tangential velocity), 
the most common case in applications; theorem \ref{th:vyneg-poscirc} assumes a Lax condition or the implied weak Lax condition \eqref{eq:weaklax} instead.

Apart from the Euler equations the non-existence of triple shocks has also been shown for the transonic small disturbance equation (TSD)
and other models (see \citep{rosales-tabak}, \citep[Appendix 2]{morawetz-potential-theory}, \citep[Theorem 1.1]{yuxi-zheng-book}, 
\citep[Theorem 2.4]{serre-hbfluidmech});
for \emph{full} (not TSD) compressible potential flow 
we find in section \ref{section:potentialflow} that pure triple shocks do exist, 
in particular for $\gisen$-law pressure with $1<\gisen<3$, including all common values (example \ref{ex:gtwo} and \ref{ex:gall}). 

In section \ref{section:twoplustwo} we show that if instead of a single Mach stem two or more lower shocks are permitted, then circulation sign control cannot be expected even for polytropic flow.

\section{Zero-circulation shock interactions}
\label{section:zerocirc}%

In $k+1$ MR (see fig.\ \ref{fig:kshocks}) multiple \defm{incident} shocks numbered $1,...,k-1$ clockwise from the negative horizontal axis 
are followed by the \defm{reflected} shock $k$, 
the contact $c$ and the \defm{Mach stem} shock $0$, all halflines meeting only in the origin. The \defm{upper} shocks $1,...,k$ separate constant states $0,...,k$, 
state $k$ between reflected shock and contact, state $c$ between contact and stem. By scaling and rotational symmetry we may take $\vv_0=(1,0)$ throughout. 
Unit shock normals $\vn$ are clockwise for the upper shocks, counterclockwise for the Mach stem; jumps $[]$ are change in $\vn$ direction.

\clearpage

\begin{figure}
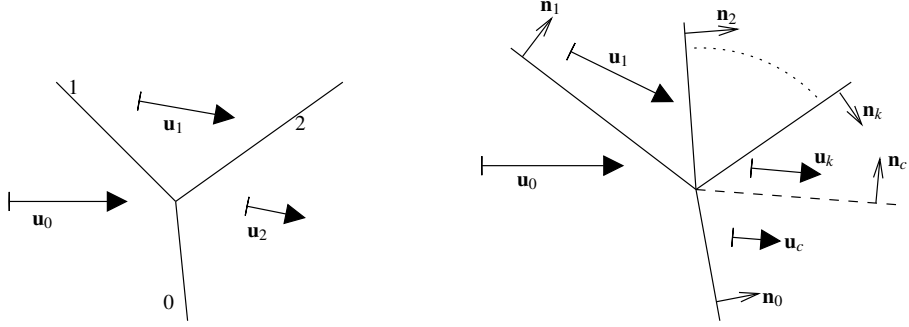

  \parbox[b]{.49\textwidth}{%
    \input{puretriple.tex}
  }\hfil%
  \parbox[b]{.49\textwidth}{%
    \input{kshocks.tex}%
  }%
  \caption{Left: pure triple shock. Right: $k+1$ MR}
  \label{fig:tripleshock}
  \label{fig:kshocks}
\end{figure}

Consider a pure triple-shock, i.e.\ $2+1$ MR without contact (fig.\ \ref{fig:tripleshock}). 
We may assume $[\pp]>0$ across each shock, by renumbering states if necessary.
Conservation of mass across discontinuities means zero jump of mass flux $\dens\vv\dotp\vn$, which is nonzero for a shock (as opposed to contact) so that
conservation of tangential momentum
\[ 0 = [\dens~\vv\dotp\vn~\vv\dotp\vt] = \dens\vv\dotp\vn[\vv\dotp\vt] \] 
implies tangential velocity $\vv\dotp\vt$ cannot jump either. Then $[\vv]$ is a shock normal, so 
\[ \big|[\vv]\big|^2=\big([\vv]\dotp\vn\big)^2 = \big([\idens\dens\vv\dotp\vn]\big)^2 = [\idens]^2 (\dens\vv\dotp\vn)^2 \] 
allows eliminating the unit normal $\vn$ in conservation of normal momentum:
\[ \subeq{[\pp]}{\eqdef\bM>0} = [-\dens(\vv\dotp\vn)^2] = [-\idens] (\dens\vv\dotp\vn)^2 = \big|\subeq{[\vv]}{\eqdef\bP}\big|^2 / \subeq{[-\idens]}{\eqdef 2\bE}.  \] 
So if we associate shocks with ``particles'' of mass $\bM$ and momentum $\bP$, then $\bE=|\bP|^2/2\bM$ is energy; $\bM=[\pp]>0$ implies $\bE>0$. 
In absence of a contact
\[ \subeq{\vv_2-\vv_0}{=\bP_0} = \subeq{\vv_2 - \vv_1}{=\bP_2} + \subeq{\vv_1-\vv_0}{=\bP_1}; \] 
similarly $\bM_0=\bM_1+\bM_2$ and $\bE_0=\bE_1+\bE_2$,
so a triple shock corresponds to an \emph{elastic} collision of two incoming particles forming a single outgoing cluster --- which is well-known to be impossible. 
\begin{theorem}
  \label{th:twoplusonetriple}
    The Euler equations do not permit pure triple shocks
    (even if we allow expansive shocks and non-conservation of energy and make no assumptions about the equation of state).
\end{theorem}
(The proof is familiar e.g.\ from kinetic gas theory: using Galilean invariance pass to the rest frame of cluster $0$ where its energy is zero so that
conservation would imply the incoming particles also have zero energy and hence zero velocity $\bv_i$. 
Mathematically, subtract from $\bE$ conservation the dot product of $\bv_0$ with $\bP$ conservation, 
then add $\half|\bv_0|^2$ times $\bM$ conservation to obtain $\half\bM_1|\bv_1-\bv_0|^2+\half\bM_2|\bv_2-\bv_0|^2=0$.)

So far we have a slight extension of Theorem 2.3 in \citep{serre-hbfluidmech}.
But the particle collision picture of our new proof readily suggests generalizations,
for example to $k+1$ MR with $k\geq 3$ (fig.\ \ref{fig:kshocks}).
Multiple particles clustering into one is possible by \emph{in}elastic collision, with decrease of kinetic energy:
\[ \idens_0-\idens_c = 2\bE_0 < 2\bE_1 + ... + 2\bE_k < \idens_0-\idens_1 + ... + \idens_{k-1} - \idens_k \]
so $\idens_c-\idens_k>0$, hence $\dens_k-\dens_c<0$. In contrast to prior work we obtain nonexistence even if we only 
require the density jump to be \emph{nonpositive} rather than zero.

It is necessary to make additional assumptions, however. To begin with we cannot distinguish 3+1 interactions from 2+2 (see fig.\ \ref{fig:fournocontacts}), which are well-known to exist,
unless we assume normal velocities are clockwise for the upper shocks (fig.\ \ref{fig:kshocks}), counterclockwise for the Mach stem, which completes our definition of $k+1$ MR.
In addition we still need $[\pp]>0$ so that (see above) $[\dens]$ and hence $\bM,\bE$ are positive, but renumbering is generally not possible for more than three states, so we need to \emph{assume} all shocks are compressive. 
\begin{theorem}
  \label{th:kplusonedensity}%
  The Euler equations do not permit a compressive $k+1$ MR with zero velocity jump $\vv_k-\vv_c$ across the contact
  unless the density jump $\dens_k-\dens_c$ is positive (in particular nonzero). 
\end{theorem}

Shocks are usually compressive, although physical examples of expansive shocks have been considered \citep{bethe,zeldovich-raizer,lambrakis-thompson-1972,colonna-guardone-nannan-2007}.
If we omit compressiveness for $k\geq 3$, then there are counterexamples: 
\begin{example}
\begin{alignat*}{9} \vv_0 &= (1,0),&\ \vv_1 &= (\frac34,-\frac14),&\ \vv_2 &= (\frac54,-\frac14),&\ \vv_3 &= (\frac{36}{37},\frac{9}{74}) = \vv_c , \\
\dens_0 &= 1,&\ \dens_1 &= 2,&\ \dens_2 &= \frac65 ,&\ \dens_3 &= \frac{148}{63} = \dens_c, \notag\\
\pp_0 &= 1,&\ \pp_1 &= \frac54,&\ \pp_2 &= \half ,&\ \pp_3 &= \frac{38}{37} = \pp_c. \end{alignat*} 
Using $\vv_i$ it is easy to calculate the shock normals $\vn_i$ and tangents and to verify each jump relation and that with $\vn_{1,2,3}$ clockwise and $\vn_0$
counterclockwise we have $\vv\dotp\vn>0$ on each side of each shock. 
The outward tangent angles of the discontinuities are
\begin{alignat*}{9} \alpha_1 &= 135^\circ ,\  \alpha_2 = 90^\circ ,\  \alpha_3 \approx 36.7 ,\  \alpha_0 \approx -167.5, \end{alignat*} 
hence shocks appear in the right order. 
Here shock 2 is non-compressive.
\end{example}
This example, like later ones, is ``kinematic'': attention is paid only to jump relations for mass and momentum,
although we emphasize that conservation of energy \emph{can} always be satisfied simply by defining appropriate $\epm_i$. 
Then $\pp_i,\epm_i$ may even arise from a (not necessarily reasonable) equation of state, except in (avoidable) cases of two different $\epm_i$ assigned to the same $(\dens_i,\pp_i)$.

Example \ref{ex:zerocircposrjump} will show we can\emph{not} rule out positive density jumps with zero circulation without making additional assumptions. 
But for 2+1 MR it is possible to prove a strong additional constraint. 
Since $[\vv]$ is a shock normal we may eliminate $\vn=[\vv]/|[\vv]|$:
\begin{alignat}{5} [\pp] = -[\dens(\vv\dotp\vn)^2] 
= -\dens\vv\dotp\vn[\vv\dotp\vn] 
= -\dens\vv\dotp[\vv]. \label{eq:puu}\end{alignat} 
Therefore
\begin{alignat*}{5} \pp_2-\pp_0 &= \pp_2-\pp_1+\pp_1-\pp_0 \\&= \dens_0\vv_0\dotp(\vv_0-\vv_1) + \dens_1\vv_1\dotp(\vv_1-\vv_2). \end{alignat*}
Across the Mach stem (we may use $\pp_2=\pp_c$ in any case and $\vv_2=\vv_c$ by assumption):
\begin{alignat*}{5} \pp_2-\pp_0 &= \pp_c-\pp_0 = \dens_0\vv_0\dotp(\vv_0-\vv_c) 
\\&= \dens_0\vv_0\dotp(\vv_0-\vv_1) + \dens_0\vv_0\dotp(\vv_1-\subeq{\vv_c}{=\vv_2}) \end{alignat*} 
Subtract the last two equations: 
\begin{alignat*}{5} 0 &= (\dens_1\vv_1-\dens_0\vv_0)\dotp(\vv_1-\vv_2). \end{alignat*} 
If shock $1$ is \emph{not} normal, then $\dens_1\vv_1-\dens_0\vv_0$ is by $[\dens\vv]\dotp\vn=0$ a \emph{nonzero} tangent to shock $1$; $\vv_1-\vv_2$ is a \emph{normal} to shock $2$, 
so the two shocks either coincide (which we excluded) or are antipodes. But antipodes implies $\vv_1-\vv_2\parallel\vv_0-\vv_1$, 
so both are $\parallel\vv_0-\vv_2$ which is normal to shock $0$. But then shock $0$ must coincide either with $1$ or $2$, contradiction:
\begin{theorem}
  \label{th:twoplusonedensity}
  For a 2+1 MR (not necessarily compressive) with zero circulation the density jump across the contact is positive and the incident shock is \emph{normal}.
\end{theorem}
The last conclusion by itself is strong enough to conflict with the circumstances of most practical applications, where incident shocks are backward.

The following example shows we cannot rule out zero circulation with positive density jump if we merely assume compressive shocks:
\begin{example} 
  \label{ex:zerocircposrjump}%
  \begin{alignat*}{5} \vv_0 &= (1,0) \csep &\vv_1 &= (\frac23,0) \csep &\vv_2 &= (\frac13,\frac14) &= \vv_c \\
  \pp_0 &= 1 \csep &\pp_1 &= \frac43 \csep &\pp_2 &= \frac53 &= \pp_c \\
  \dens_0 &= 1 \csep &\dens_1 &= \frac32 \csep &\dens_2 &= \frac{48}{7} \csep \dens_c &= \frac{96}{23} \end{alignat*} 
  Clearly $\dens_2>\dens_c$, as implied by Theorem \ref{th:kplusonedensity}. Discontinuity outward tangent angles:
  \begin{alignat*}{5} \alpha_1 &= 90^\circ \csep \alpha_2 \approx 53^\circ \csep \alpha_c \approx 37^\circ \csep \alpha_0 \approx -111^\circ \end{alignat*} 
\end{example}

For k+1 MR with $k\geq 3$ there is more freedom and we can no longer prove the (first) incident shock is normal;
numerical examples can be obtained by splitting shock $1$ of a 2+1 MR into a backward and forward half 
(as we prove in Theorem \ref{th:circ-backward} it is not possible to have all-backward incident shocks):
\begin{example} 
  \label{ex:zerocircposrjumpthree}%
  \begin{alignat*}{5} 
    \vv_0 &= (1,0) ,\quad & \pp_0 &= 1 ,\quad & \dens_0 &= 1 ,\quad \\
    \vv_1 &= (\frac9{10},-\frac1{30}) ,\quad & \pp_1 &= \frac{11}{10} ,\quad & \dens_1 &= \frac98 ,\quad \\ 
    \vv_2 &= (\frac7{10},\frac1{30}) ,\quad & \pp_2 &= \frac{261}{200} ,\quad & \dens_2 &= \frac{369}{248} \\
    \vv_3 &= (\frac35,\frac{134}{615}) ,\quad &\pp_3 &= \frac{7}{5} ,  &\dens_3 &= \frac{287451}{59876} \\
    \vv_c &= \vv_3 ,\quad & \pp_c &= \pp_3 ,\quad & \dens_c &= \frac{75645}{36409}  
  \end{alignat*} 
  Discontinuity outward tangent angles:
  \begin{alignat*}{5} \alpha_1 &\approx 108^\circ ,\quad \alpha_2 \approx 72^\circ ,\quad \alpha_3 \approx 28^\circ ,\quad \alpha_c \approx 20^\circ ,\quad \alpha_0 \approx -119^\circ \end{alignat*} 
\end{example}

\section{Sign of circulation}
\label{section:circulation}

\subsection{Preliminaries}

Having exhausted the question of zero circulation we consider kinematic arguments for the 
\emph{sign} of circulation. Most commonly circulation on the contact is observed to be \emph{clockwise}, with $\vv_k,\vv_c$ all in the open right halfplane and $|\vv_k|>|\vv_c|$, 
hence \emph{positive} velocity jump $|\vv_k|-|\vv_c|$. 
Although compressiveness went a long way towards ruling out circulation becoming zero, ruling out negative requires additional assumptions, 
as shown by the following example (which is a small perturbation of example \ref{ex:zerocircposrjump}):
\begin{example}
  \label{ex:negcircposrjump}%
  \begin{alignat*}{5} \vv_0 &= (1,0) ,\quad \vv_1 = (\frac23,\frac19) ,\quad \vv_2 = (\frac13,\frac13) ,\quad \vv_c = \frac{18}{17} \vv_2 \\
    \pp_0 &= 1 ,\quad \pp_1 = \frac43 ,\quad \pp_2 = \frac{28}{17} = \pp_c \\
    \dens_0 &= 1 ,\quad \dens_1 = \frac{27}{17} ,\quad \dens_2 = \frac{144}{17} ,\quad \dens_c = \frac{187}{30} 
  \end{alignat*} 
  Clearly $|\vv_2|-|\vv_c|<0$, i.e.\ negative circulation.
  Angles of outward discontinuity tangents:
  \begin{alignat*}{5} \alpha_1 &\approx 71.6^\circ ,\quad \alpha_2 \approx 56.3^\circ ,\quad \alpha_c = 45^\circ ,\quad \alpha_0 \approx -118.6^\circ \end{alignat*} 
\end{example}
To prove a conclusive result about circulation sign we need to add additional assumptions. 
We first need to rule out ``exotic'' MR; for the proof of the following proposition see the Appendix.
\newcommand{\propmanya}{%
  Consider $k+1$ MR with compressive shocks. \\
  (a) $\vv_c$ is in the open right halfplane.\\
  (b) The Mach stem is in the open lower halfplane. \\
  (c) If shocks $1,...,i$ are in the open upper left quadrant, then they are backward and produce $\vv_1,...,\vv_i$ in the open lower right quadrant;
  angular velocity is clockwise everywhere in the open upper left quadrant.\\
  (d) $\vv_c$, $\vv_k$ and contact point exactly in the same direction. \\
  (e) The angle between consecutive upper shocks is less than $180^\circ$. 
}
\begin{proposition}%
  \label{prop:manya}%
  \propmanya%
\end{proposition}

By (a)+(d) contact, $\vv_k$ and $\vv_c$ are in the right halfplane and parallel, 
so 
\begin{alignat*}{5} |\vv_k|-|\vv_c|>0 \qeq \text{clockwise circulation}. \end{alignat*}
Clockwise circulation is the ``standard'' case.

\subsection{Shock fan reduction by ``pressure sliding''}

\begin{figure}
  \input{plines.tex}
  \caption{Given $\vv_{k-1},\dens_{k-1},\pp_{k-1}$, the level sets of $\vv_k$ with equal $\pp_k$ are lines perpendicular to $\vv_{k-1}$;
    the $\dens_k$ level sets are circles (fig.\ \ref{fig:rhocircle})}
  \label{fig:plines}
\end{figure}
\begin{figure}
  \input{lowerchain.tex}
  \caption{A fan of shocks is reduced by combining the last two without changing the final pressure. The $\vv_{i-1}-\vv_i$ are downstream shock normals.}
  \label{fig:lowerchain}
\end{figure}
Pressure $\pp$ does not jump across contacts, so it is a good quantity for controlling the relationship between Mach stem and incident/reflected shocks. 
We use that the shock relations imply (cf.\ \eqref{eq:puu})
\begin{alignat}{5} \pp_k - \pp_{k-1} = \dens_{k-1}\vv_{k-1}\dotp(\vv_{k-1}-\vv_k) \label{eq:pukmo} \end{alignat} 
for shock $k$ and the analogous relation for the other shocks (note that the terminal density $\dens_k$ is absent).
This $\pp$ relation allows an intuitive way of thinking about downstream pressure (fig.\ \ref{fig:plines}): 
for fixed $\vv_{k-1},\dens_{k-1},\pp_{k-1}$ the level sets of $\vv_k$ with same $\pp_k$ are straight lines perpendicular to $\vv_{k-1}$.

This immediately suggests a reduction technique (fig.\ \ref{fig:lowerchain}): 
we may collapse shocks $k$ and $k-1$ into a single new ``shock'' without changing the terminal pressure $\pp_k$ by 
shifting $\vv_k$ along its level set onto the line through $\vv_{k-1}$ and $\vv_{k-2}$, to a new terminal velocity $\tilde\vv_{k-1}$. 
Then
\begin{alignat*}{1}
  \tilde\pp_{k-1} - \pp_{k-1} &= \dens_{k-1}\vv_{k-1} \dotp (\vv_{k-1}-\tilde\vv_{k-1}) 
  \\&= \dens_{k-2}\vv_{k-2} \dotp (\vv_{k-1}-\tilde\vv_{k-1}) 
\end{alignat*}
because $\dens_{k-1}\vv_{k-1}-\dens_{k-2}\vv_{k-2}$ is tangential to shock $k-1$, hence perpendicular to $\vv_{k-1}-\vv_{k-2}$ which is parallel to $\vv_{k-1}-\tilde\vv_{k-1}$ by choice of $\tilde\vv_{k-1}$. By adding the relation for shock $k-1$,
\begin{alignat*}{1}
  \pp_{k-1} - \pp_{k-2} &=\dens_{k-2}\vv_{k-2} \dotp (\vv_{k-2}-\vv_{k-1}),
\end{alignat*}
we verify that 
\begin{alignat*}{1}
  \tilde\pp_{k-1} - \pp_{k-2}
  &= \dens_{k-2}\vv_{k-2} \dotp (\vv_{k-2}-\tilde\vv_{k-1})
\end{alignat*}
so that the new combined ``shock'' satisfies the same $\pp$ relation. (We put ``shock'' in quotes since $\tilde\dens_{k-1}$ may be negative or undefined, which is harmless because like the original terminal density $\dens_k$ we avoid using it anywhere.)
Finally $\tilde\pp_{k-1}=\pp_k>\pp_{k-1}>\pp_{k-2}$ means $\vv_{k-1}$ is between $\tilde\vv_{k-1}$ and $\vv_{k-2}$ so that $\tilde\vv_{k-1}-\vv_{k-2}$ retains the same direction as $\vv_{k-1}-\vv_{k-2}$.

We repeat this reduction step with $\tilde\vv_{k-1},\vv_{k-2},\vv_{k-3}$ taking the place of $\vv_k,\vv_{k-1},\vv_{k-2}$, then with $\tilde\vv_{k-2},\vv_{k-3},\vv_{k-4}$ etc.; 
  after finitely many steps we have reduced the upper shocks to a single new ``shock'' connecting $\vv_0$ to $\tilde\vv_1$, with
\begin{alignat*}{5} \tilde\pp_1 = \pp_0 + \dens_0\vv_0\cdot(\vv_0-\tilde\vv_1) \overset{\vv_0=(1,0)}{=} \pp_0 + \dens_0(1-\tilde\ve^x_1), \end{alignat*} 
and by $\tilde\pp_1=...=\tilde\pp_{k-1}=\pp_k=\pp_c$ 
the remaining relation 
\begin{alignat*}{1}
  \pp_c=\pp_0+\dens_0\vv_0\dotp(\vv_0-\vv_c)=\pp_0+\dens_0(1-\vx_c)
\end{alignat*}
across the Mach stem yields $\tilde \ve^x_1=\vx_c$. 
So if we can control the change from $\ve^x_k$ to $\tilde\ve^x_1$ during the reduction steps we have control over $\vv_k-\vv_c$.

To this end we calculate that $\ve^x$ can only decrease: in the first reduction step 
\begin{alignat}{5} \tilde\vv_{k-1} = \vv_k - t \vv_{k-1}^\perp \label{eq:tdef} \end{alignat}
($\perp$ rotation counterclockwise by $90^\circ$) which is on the line through $\vv_{k-1},\vv_{k-2}$ if
\begin{alignat*}{5} 0 &\overset!= (\tilde\vv_{k-1}-\vv_{k-1})^\perp\dotp(\vv_{k-2}-\vv_{k-1}), \end{alignat*}
so necessarily
\begin{alignat*}{5} t &= \frac{(\vv_{k-1}-\vv_k)^\perp\dotp(\vv_{k-2}-\vv_{k-1})}{\vv_{k-1}\dotp(\vv_{k-2}-\vv_{k-1})}. \end{alignat*}
The denominator is positive since $\vv_{k-2}-\vv_{k-1}$ is a downstream normal, by compressiveness.
The numerator is positive since by Proposition \ref{prop:manya}(e) consecutive upper shocks are at angles less than $180^\circ$ so that 
the outward tangent $(\vv_{k-1}-\vv_k)^\perp$ of the last shock has positive dot product with the downstream normal $\vv_{k-2}-\vv_{k-1}$ of the previous. 
(This remains true in later reduction steps since as mentioned $\tilde\vv_{k-1}-\vv_{k-2}$ has the same direction as $\vv_{k-1}-\vv_{k-2}$, etc.)
So $t>0$, and since \eqref{eq:tdef} shows
\begin{alignat*}{5} \tilde\ve^x_{k-1} = \vx_k + t \vy_{k-1}, \end{alignat*} 
we find that $\vy_{k-1}\leq 0$ would imply $\tilde\ve^x_{k-1}\leq\vx_k$.
In fact $\vx$ \emph{must} decrease during reduction if all $\vy_i$ are nonpositive and least one is negative.
Then $\vx_k>\tilde\ve^x_1$ which was shown $=\vx_c$ above, so we have proven (recall $\vv_k,\vv_c$ have the same direction, by Prop.\ \ref{prop:manya}(d)):
\begin{proposition}
  \label{prop:backward-p}%
  Consider an MR with compressive shocks and vertical velocities $\vy_1,...,\vy_{k-1}$ between the upper shocks nonpositive and not all zero.
  Then $|\vv_k|>|\vv_c|$ (clockwise circulation). 
\end{proposition}
(``Not all zero'' is automatic for $k\geq 3$, but not $k=2$ case as example \ref{ex:zerocircposrjump} shows). 
By Proposition \ref{prop:manya}(c) we immediately obtain:
\begin{theorem}
  \label{th:circ-backward}%
  If in a compressive $k+1$ MR every incident shock is backward, then the circulation on the contact is clockwise. 
\end{theorem}
This theorem is especially useful because in many concrete applications the incident shocks are naturally backward, propagating downstream towards an interaction point.

\subsection{Weak Lax condition}

Although examples \ref{ex:zerocircposrjump} and \ref{ex:negcircposrjump} look natural, they have one feature that 
--- as we show now --- cannot be improved: $\vv_1\dotp\vn_1 > \vv_1\dotp\vn_2$. 
So the incident shock has downstream normal velocity higher than the upstream normal velocity at the following reflected shock, 
but that would be incompatible with the Lax condition since the sound speed $\csnd$ is constant between the shocks:
downstream of shock 1 we need
\begin{alignat*}{5} \csnd_1 > \vv_1\dotp\vn_1, \end{alignat*} 
and upstream of shock 2 we need
\begin{alignat*}{5} \vv_1\dotp\vn_2 > \csnd_1, \end{alignat*} 
which combines to a \defm{weak Lax condition}: 
\begin{alignat}{5} \vv_1\dotp\vn_1 < \vv_1\dotp\vn_2. \label{eq:weaklax}\end{alignat} 
\begin{theorem}
  \label{th:vyneg-poscirc}%
  If the upper shocks in a compressible $k+1$ MR satisfy the weak Lax condition \eqref{eq:weaklax} (in particular if they satisfy the standard Lax condition), then the incident shocks are backward and 
  Proposition \ref{prop:backward-p}(b) applies, showing in particular circulation on the contact is clockwise.
\end{theorem}
\begin{proof}
  Assume the incident shocks are not all backward. 
  Of the forward ones consider the first in clockwise direction. 
  The weak Lax condition requires the clockwise angular velocity on the upstream side of the following shock to be strictly larger. 
  But a forward shock has \emph{outward} radial downstream velocity, 
  which means clockwise angular velocity \emph{de}creases when moving in clockwise direction; the velocity becomes outward purely radial velocity before the next shock. 
  From there the velocity vector would have to rotate more than $180^\circ$ 
  to turn inward purely radial and then return to positive and increasing clockwise angular velocity. This requires a sector of equal clockwise angle to be traversed, 
  contradicting
  Proposition \ref{prop:manya}(e).
  Hence all incident shocks must be backward, so theorem \ref{th:circ-backward} applies. 
\end{proof}

\subsection{Density jump}
\label{section:thermo}

Finally we wonder whether the weak Lax condition can guarantee that the density jump $\dens_k-\dens_c$ is positive, 
as it is for polytropic $\gisen>1$ equation of state for example,
but that is not generally true:
\begin{example}
\begin{alignat*}{5} \vv_0 &= (1,0) ,\quad &\vv_1 &= (\frac89,-\frac19) ,\quad &\vv_2 &= (\frac14,\frac15) ,\quad &\vv_c &= \frac{16}{35} \vv_2 \\
\pp_0 &= 1 ,\quad &\pp_1 &= \frac{10}{9} ,\quad &\pp_2 &= \frac{66}{35} & \pp_c &= \pp_2 \\
\dens_0 &= 1 ,\quad &\dens_1 &= \frac97 ,\quad &\dens_2 &= \frac{19520}{2457} ,\quad &\dens_c &= \frac{27125}{2844} 
\end{alignat*} 
with outward tangents of discontinuities at angles
\begin{alignat*}{5} \alpha_1 &\approx 135^\circ \csep \alpha_2 \approx 64^\circ \csep \alpha_c \approx 39^\circ \csep -96^\circ . \end{alignat*} 
Clearly the density jump $\dens_2-\dens_c$ is \emph{negative}. Moreover the shocks are obviously compressive, and
\begin{alignat*}{5} \vv_1\dotp\vn_1 = 0.55... < 0.84... < \vv_1\dotp\vn_2 \end{alignat*} 
so that the weak Lax condition is satisfied. (In fact $(\pp_1-\pp_0)/(\idens_1-\idens_0) > (\pp_2-\pp_1)/(\idens_2-\idens_1)$, so 
there is no obstacle to defining an equation of state with genuine nonlinearity $(\partial^2\pp/\partial\idens^2)_{\spm}>0$.)
\end{example}
(It is also possible to give examples \emph{without} $\dens$ jump.)

To keep the density jump positive, and to derive information about the across-contact jumps of other quantities, 
it is finally necessary to employ full thermodynamics:

Across full Euler (steady) shocks we note the Bernoulli relation: 
with $\hpm=\epm+\pp\idens$ enthalpy per mass, 
\begin{alignat*}{5} [\half|\vv|^2 + \hpm] = 0; \end{alignat*} 
summing over all shocks we obtain the corresponding relation
\begin{alignat*}{5} \half|\vv_k|^2 + \hpm_k = \half|\vv_c|^2 + \hpm_c \end{alignat*} 
across the contact. If circulation is clockwise, i.e.\ $|\vv_c|^2<|\vv_k|^2$, then necessarily
\begin{alignat*}{5} \hpm_c > \hpm_k . \end{alignat*} 
On the other hand $\pp_c=\pp_k$, so we may use $\hpm,\pp$ to study the jump for any other thermodynamic variable,
by studying their behaviour on the isobars (curves of constant $\pp$). 
\begin{alignat*}{5} d\hpm = 
\Temp~d\spm + \idens~d\pp, \end{alignat*} 
so $ \dati{\hpm}{\spm}{\pp} = \Temp > 0 $
which immediately implies 
\begin{alignat*}{5} \spm_c > \spm_k \end{alignat*} 
for positive circulation, the reverse for negative circulation and zero entropy jump if there is no velocity jump either.
Similar observations entered the thermodynamic arguments of \cite{henderson-menikoff}.

However, the jump of $\idens$ cannot be controlled by the jump of $\spm$ or other variables in such a simple manner: 
\begin{alignat*}{5} \dat\idens\spm\pp = - \dat\pp\spm\idens \dat\pp\idens\spm^{-1} ; \end{alignat*} 
although sound speed
\begin{alignat*}{5} \csnd^2 = - \idens^2 \dat\pp\idens\spm  \end{alignat*} 
may of course reasonably be assumed to be positive,
there is no \defm{fundamental} reason why $\dati\pp\spm\idens$ 
should be positive. 
In fact it may well be negative: some calculation shows 
\begin{alignat*}{5} \dat\pp\spm\idens = (\epm_{\idens\idens}\epm_{\spm\spm}-\epm_{\idens\spm}^2) \dat\idens\Temp\pp . \end{alignat*} 
The determinant may be assumed to be positive by the \defm{thermodynamic inequalities}, i.e.\ positive definite Hessian of $\epm=\epm(\idens,\spm)$.
But the last factor is the coefficient of thermal expansion at constant pressure, 
which is well-known to be negative for some materials in some temperature and pressure regimes, especially near phase transitions. 
Besides, this section is based on an assumption of equilibrium which cannot be made in case of metastability, say for chemically reacting mixtures.

\section{Potential flow}
\label{section:potentialflow}

\subsection{Background}

We recall briefly how and where compressible potential flow arises; some readers may wish to skip to the following section.

The Euler equations 
\begin{alignat}{5} 0 &= \dens_t + \ndiv(\dens\vv), \label{eq:masst}\\
0 &= (\dens\vv)_t + \ndiv(\dens\vv\otensor\vv) + \nabla\pp, \label{eq:momt}\\
0 &= (\dens(\epm+|\vv|^2/2))_t + \ndiv((\dens(\epm+|\vv|^2/2)+\pp)\vv) \end{alignat} 
have smooth solutions satisfying a transport law 
\begin{alignat*}{5} 0 &= \spm_t + \vv\dotp\nabla\spm, \end{alignat*} 
so that $\spm$ remains constant if it was initially. For constant $\spm$ we obtain the \defm{isentropic Euler equations} 
(sometimes called ``adiabatic'' or ``barotropic'', although those names are also used for related systems)
\begin{alignat}{5} 0 &= \dens_t + \ndiv(\dens\vv), \label{eq:isenmass} \\
0 &= (\dens\vv)_t + \ndiv(\dens\vv\otensor\vv) + \nabla\pp. \label{eq:isenmom}\end{alignat}
with $\pp=\ppf(\dens)$.
\emph{Weak} (distributional) solutions of isentropic Euler are not weak solutions of the full Euler equations, but shocks with small strength are close approximations of full Euler shocks. 
Every theorem from sections \ref{section:zerocirc} and \ref{section:circulation} 
applies directly to isentropic Euler since we avoided using the conservation of energy equation. 

From \eqref{eq:isenmass} and \eqref{eq:isenmom} with $\dens^{-1}\dati\ppf\dens\spm=\dati\hpmf\dens\spm$ 
we obtain 
\begin{alignat}{5} 0 &= \vv_t + \vv\dotp\nabla\vv + \nabla\hpm. \label{eq:vvt}\end{alignat} 
Taking the curl we see that smooth (and vacuum-free) isentropic Euler solutions also satisfy a transport law for vorticity $\vort=\ncurl\vv$:
\begin{alignat*}{5} 0 &= (\frac{\vort}{\dens})_t + \vv\dotp\nabla\frac{\vort}{\dens}, \end{alignat*} 
so if $\vort$ is zero initially it remains so for all time. Then $\vv = \nabla\vpot$
for a scalar \defm{velocity potential} $\vpot$. 
This turns \eqref{eq:vvt} into 
\begin{alignat*}{5} 0 &= \nabla\vpot_t + \nabla\vpot\dotp\nabla^2\vpot + \nabla\hpm = \nabla(\vpot_t + \half|\nabla\vpot|^2 + \hpm), \end{alignat*}
resulting in the the \defm{Bernoulli relation}
\begin{alignat*}{5} \vpot_t + \half|\nabla\vpot|^2 + \hpmf(\dens) = 0 \end{alignat*}
where the right-hand side is made zero by adding some irrelevant $f(t)$ to $\vpot$. \defm{Potential flow} is the system
\begin{alignat}{5} 
\dens &= \hpmf^{-1}(-\vpot_t -\half|\nabla\vpot|^2)  , \\
0 &= \dens_t + \ndiv(\dens\vv) .
\label{eq:densvpot}\end{alignat} 
Smooth potential flow solutions are full Euler solutions,
but again weak solutions (which require that $\vpot$ is continuous, but not necessarily its derivatives)
need not be, although shocks of small strength are still close approximations.
Admissibility conditions are discussed in \cite{osher-hafez-whitlow,elling-potfadm} 
and references therein. 

Potential flow also arises in another important context: consider the nonlinear Schr\"odinger equation
\begin{alignat}{5} i\hnls\wavefn_t = -\frac{\hnls^2}{2}\Lap\wavefn + \Vpot\wavefn, \label{eq:nls}\end{alignat} 
where $\Vpot$ could be a function of $t,\xx$ or also $|\wavefn|$, for example
\begin{alignat*}{5} \Vpot &= |\wavefn|^{\pV-1}. \end{alignat*} 
Substitute the usual ansatz 
\begin{alignat*}{5} \wavefn = \amp e^{i\phs/\hnls} \end{alignat*} 
with real \defm{amplitude} $\amp$ and \defm{phase} $\phs$ to obtain
as real part the usual eikonal-type equation
\begin{alignat}{5} 0 = \phs_t + \half  |\nabla\phs|^2 + \Vpot - \half \hnls^2\frac{\Lap\amp}{\amp} \label{eq:sreal}\end{alignat}
and as imaginary part
\begin{alignat*}{5} 0 = \amp_t + \nabla\amp\dotp\nabla\phs + \half \amp \Lap\phs . \end{alignat*}
The role of the latter becomes clearer with $\amp=\sqrt{\dens}$ and $\phs=\vpot$, yielding
\begin{alignat*}{5} 0 = \dens_t + \dens\Lap\phs + \nabla\dens\dotp\nabla\vpot = \dens_t + \ndiv(\dens\nabla\vpot), \end{alignat*}
which is exactly the continuity equation, a conservation law that maintains $\int|\wavefn|^2d\xx=\const$ in time. 
Moreover $\vv=\nabla\phs$ is the \defm{group velocity} from the dispersion relation of \eqref{eq:nls};
taking $\nabla$ of the real part \eqref{eq:sreal} yields 
\begin{alignat}{5} 0 = \vv_t + \subeq{\nabla\vv\dotp\vv}{=(\vv\dotp\nabla)\vv} + \nabla\Vpot - \hnls^2\nabla\frac{\Lap\amp}{2\amp} . \label{eq:vvvan}\end{alignat} 
This equivalent representation of the Schr\"odinger equation also goes by \defm{Madelung fluid} (see \cite{madelung-1927}).
In the $\hnls\dnconv 0$ limit \eqref{eq:sreal} formally converges to the Bernoulli relation for potential flow; 
$\Vpot=C\amp^{\pV-1}$ plays the role of $\hpm=C\dens^{\gisen-1}$, with $2(\gisen-1)=\pV-1$. 
While the Euler equations generally arise as a vanishing dissipation limit, 
the vanishing term in \eqref{eq:vvvan} is a third-order operator with dispersive effects.
NLS also arises in nonlinear optics and water waves, for example; for discussion of Mach reflection in water wave models see \citep{yue-mei-1980,kodama-yeh-2016}.

The shock waves of potential flow satisfy the familiar mass conservation
\begin{alignat*}{5} [\dens\vv\dotp\vn] = 0 \end{alignat*} 
as well as the Bernoulli relation (equivalent to conservation of energy)
\begin{alignat*}{5} [\Hpm] = 0 \end{alignat*} 
with \defm{total enthalpy} $\Hpm=\hpm+\half|\vv|^2$, 
and conservation of tangential momentum in the form
\begin{alignat*}{5} [\vv\dotp\vt] = 0 \end{alignat*} 
(which follows from $[\vpot]=0$). 
However, normal momentum is \emph{not} conserved. (It is approximately conserved for weak shocks; 
in the Schr\"odinger context $\dens\nabla\vpot$ is only part of the momentum density which has an additional term $\vpot\nabla\dens$ with $O(\hnls)$ coefficient
that can be neglected only if $\nabla\dens$ does not become large.)

\subsection{Euler vs.\ potential flow}
\label{section:euler-potf}%

Assume we have already found values for the constant $\vv$ and $\dens$ between the shocks so that mass and tangential momentum conservation are satisfied across the shocks. 
To construct ``kinematic'' MR for full or isentropic Euler flow we still need to find pressures $\pp$ to satisfy conservation of normal momentum, 
\begin{alignat*}{5} 0 = [\dens(\vv\dotp\vn)^2+\pp], \end{alignat*} 
at each shock; summation over all shocks yields
the \defm{pressure compatibility relation} 
\begin{alignat*}{5} 0 &= \pp_k - \pp_c = \pp_k - \pp_{k-1} + ... + \pp_1 - \pp_0 + \pp_0 - \pp_c 
\\&= \sum_{i=0}^{k-1} \dens_i\vv_i\dotp(\vv_i-\vv_{i+1}) - \dens_0\vv_0\dotp(\vv_0-\vv_c) \end{alignat*} 
(where we eliminated normals as in \eqref{eq:puu}).
This relation puts an additional scalar constraint on $\vv,\dens$; if satisfied we have one solution for each choice of the free $\pp_0$. 

In contrast, for \emph{potential flow} we have to pick enthalpies $\hpm$ satisfying the Bernoulli relation at each shock:
\begin{alignat*}{5} 0 = [\hpm + \half|\vv|^2]. \end{alignat*} 
But here summation over all shocks yields a compatibility relation
\begin{alignat*}{5} \hpm_k - \hpm_c = \half ( |\vv_c|^2 - |\vv_k|^2 ). \end{alignat*} 
If we chose velocities so that $\vv_c=\vv_k$, then the right-hand side and therefore the left-hand side are automatically zero, 
but for potential flow $\hpm$ is linked one-to-one to $\dens$, so $\dens_k=\dens_c$ is \emph{automatic} now and so is $\pp_k=\pp_c$. 

Thus the pressure compatibility relation, as well as zero density jump, are \emph{implied} for potential flow, 
in contrast to Euler flow where zero circulation and zero density jump are two independent relations. 

Given $\dens_i,\hpm_i$, it suffices to interpolate the pairs to a function $\hpmf$ and to calculate a pressure function $\pp$
by the standard relation $\ppf_\dens=\dens\hpmf_\dens$ (entropy is constant).

\subsection{Density circles}
\label{section:denscircle}%

  \begin{figure}
      \input{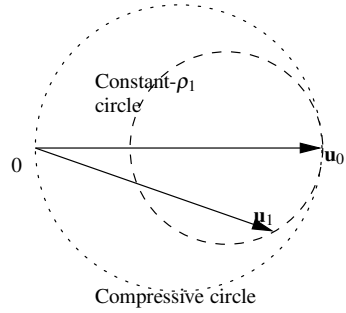}
      \caption{Small circle: given $\vv_0,\dens_0$ the level sets of $\vv_1$ with same $\dens_1$ are circles ($\pp_1$ level sets are lines (fig.\ \ref{fig:plines}));
        large circle: $\dens_1=\infty$; any compressive shock $\vv_1$ must be in the interior.}
      \label{fig:rhocircle}
  \end{figure}
  \begin{figure}
      \input{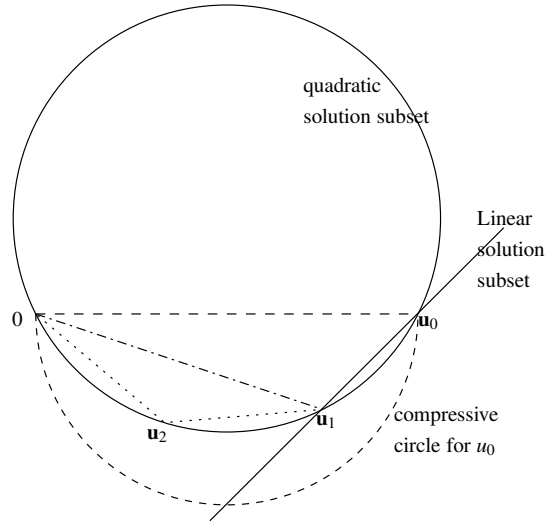}
      \caption{Solid circle: set of possible $\vv_2$ forming a potential flow triple shock with given $\vv_0,\vv_1$.}
      \label{fig:ueq}
  \end{figure}
  Hence it suffices to find values $\vv,\dens$ satisfying mass and tangential momentum conservation. 
But we only need to work for the former since the latter is automatically satisfied if we define shock normals (and locations) using $[\vv]$.

To this end we consider a given $\vv_0$ and try to find which $\vv_1$ achieve a certain desired compression ratio $\drat=\dens_1/\dens_0>0$ (see fig. \ref{fig:rhocircle}). Conservation of mass:
\begin{alignat*}{5} 0 = [\dens\vv\dotp\vn] \qeq \drat = \frac{\dens_1}{\dens_0} = \frac{\vv_0\dotp(\vv_0-\vv_1)}{\vv_1\dotp(\vv_0-\vv_1)}. \end{alignat*} 
This relation is essentially a quadratic equation for $\vv_1$, with quadratic part $|\vv_1|^2$, hence solution set essentially a circle. 
Obviously the circle must be mirror-symmetric across the $\vv_0$ axis, contain $\vv_0$ itself (corresponding to a vanishing shock), 
as well as the $\vv_1$ corresponding to a normal shock. 
So:
\begin{proposition}
  Given $\dens_0,\vv_0$ the level set of $\vv_1$ yielding same $\dens_1$ (for an Euler or potential flow shock) form a circle (see fig.\ \ref{fig:rhocircle}) with diameter the line segment
  from $\vv_0$ to $\dens_0\vv_0/\dens_1$ (the $\vv_1$ for a normal shock).
\end{proposition}

\subsection{Triple-shock circles}

Consider given $\vv_0,\vv_1$ (not multiples of each other) and $\dens_0>0$ 
and and ask which $\vv_2=\vv_k$ are such that the downstream density $\dens_2$ of shock $2$ equals the downstream density $\dens_c$ of shock $0$ (fig.\ \ref{fig:tripleshock} and\ \ref{fig:ueq}).
I.e.\ we want to solve 
\begin{alignat}{5} 1 = \frac{\dens_c}{\dens_2} = \frac{\dens_c}{\dens_0}~\frac{\dens_0}{\dens_1}~\frac{\dens_1}{\dens_2}
&= \frac{\vv_0\dotp\vn_0}{\vv_2\dotp\vn_0}~\frac{\vv_1\dotp\vn_1}{\vv_0\dotp\vn_1}~\frac{\vv_2\dotp\vn_2}{\vv_1\dotp\vn_2}. \label{eq:nn012}\end{alignat} 
Using again $\vn\parallel[\vv]$ across shocks this reduces to
\begin{alignat}{5} 1 = \frac{\vv_0\dotp(\vv_0-\vv_2)}{\vv_2\dotp(\vv_0-\vv_2)}~\frac{\vv_1\dotp(\vv_0-\vv_1)}{\vv_0\dotp(\vv_0-\vv_1)}~\frac{\vv_2\dotp(\vv_1-\vv_2)}{\vv_1\dotp(\vv_1-\vv_2)}. \label{eq:uu012}\end{alignat} 
This is essentially a cubic equation in $\vx_2$ and in $\vy_2$, which can of course be solved in brute-force fashion using Cardano formulas, 
but a more elegant argument inflicts less pain on the reader:

First, the $\vv_2$ solution set has one ``degenerate'' branch: $\vv_2$ collinear with $\vv_1,\vv_0$, corresponding to parallel shocks, i.e.\ same (or antipodal) normals $\vn_i$ 
so that \eqref{eq:nn012} is trivially satisfied.
Using this the cubic equation can be reduced to an easier quadratic one, whose solution set could be an ellipse, hyperbola, empty etc., but another trick avoids any pain
in showing it is a circle:

Perform an inversion $\zz=\vv/|\vv|^2$ which maps circles into circles (with lines regarded as extremal circles).
\eqref{eq:uu012} turns into \emph{itself} except for $\vv=\zz/|\zz|^2$ replaced by $\zz$ because
each $|\zz_i|$ and $|\zz_i-\zz_j|$ appears the same number of times in  numerator and denominator, cancelling completely. 

The $\zz$ form has the same trivial solution subset, i.e.\ $\zz_2$ on the line through $\zz_0,\zz_1$, but this line does not pass through the origin since we assumed $\vv_0,\vv_1$ and hence $\zz_0,\zz_1$ are linearly independent. Hence mapping this $\zz$ plane line back to the $\vv$ plane yields a \emph{proper} circle which contains $0$ (the $\zz$ line, being a line, ``contained infinity''). 
So we have found a second solution branch, the (unique) circle through $\vv_0,\vv_1$ and $0$ (see fig.\ \ref{fig:ueq}). 
Since circles satisfy a quadratic equation, and since we already found a trivial line subsolution set, the degree $3$ of equation \eqref{eq:uu012} does not permit additional solutions. In summary:
\begin{theorem}
  Nontrivial 2+1 MR of compressible (full) potential flow correspond one-to-one to distinct $\vv_0,\vv_1,\vv_2$ on a circle through the origin (see fig.\ \ref{fig:ueq}).
\end{theorem}

By choosing $\vv_3,\vv_4,...$ on the circle $k+1$ MR can be constructed without additional effort, 
but of course for $k\geq 3$ this yields only a subvariety of all possible nontrivial MR. 
(In particular these MR have the special property that any pair of consecutive upper shocks can be turned into a $2+1$ MR by adding a new Mach stem, 
which is not possible for arbitrary $k+1$ MR.)

\subsection{Numerical examples}

As for the Euler MR problem there is of course no guarantee that the $\dens_i,\vv_i$ combinations and resulting $\hpm_i$ correspond to an equation of state that 
is at all physically reasonable. However, even \emph{polytropic} pressure functions 
with standard isentropic coefficients $\gisen$ permit examples of potential flow triple points. 
$\ppf(\dens)=\const\cdot\dens^\gisen$ (with $\gisen\neq1$) is via $\ppf_\dens/\dens=\hpmf_\dens$ equivalent to
\begin{alignat*}{5} \hpmf(\dens) = \frac{C}{2} \dens^{\gisen-1} ; \end{alignat*} 
the jump relation $\hpm_0+|\vv_0|^2/2=\hpm_1+|\vv_1|^2/2$ yields
\begin{alignat*}{5} C &= \frac{|\vv_0|^2-|\vv_1|^2}{\dens_1^{\gisen-1}-\dens_0^{\gisen-1}} \end{alignat*} 
and since the same holds with $0,1$ replaced by $1,2$ we see 
\begin{alignat*}{5} \frac{|\vv_0|^2-|\vv_1|^2}{\dens_1^{\gisen-1}-\dens_0^{\gisen-1}} &= \frac{|\vv_1|^2-|\vv_2|^2}{\dens_2^{\gisen-1}-\dens_1^{\gisen-1}}. \end{alignat*} 
We substitute $\dens_1,\dens_2$ by functions of $\vv_0,\vv_1,\vv_2,\dens_0$; now given $\vv_0,\vv_1$ we can vary $\vv_2$ along the triple shock circle and solve numerically for $\gisen$.
For special $\gisen$ we can fix them instead and solve for $\vv_2$:

\begin{example}
  \label{ex:gtwo}%
For $\gisen=2$ solving for $\vv_2$ from $\vv_1,\vv_0$ is a quartic equation with some clean numbers solutions, e.g.
\begin{alignat*}{5} \vv_0 &= (1,0) \csep & \vv_1 &= (\frac{14}{15},-\frac{2}{15}) \csep & \vv_2 &= (\frac{2}{3},-\frac13) \quad, \\
\dens_0 &= 1 \csep & \dens_1 &= \frac32 \csep & \dens_2 &= 3\quad, \\
\hpm_0 &= \frac19 \csep &\hpm_1 &= \frac16 \csep &\hpm_2 &= \frac13; \end{alignat*}
the corresponding outward shock tangents are 
\begin{alignat*}{5} \alpha_1 &\approx 153.4^\circ \csep & \alpha_2 &= 126.9^\circ \csep & \alpha_0 &\approx -45^\circ \end{alignat*} 
\end{example}

Generally numerical solutions for a variety of $\gisen$ are easily obtained:
\begin{example} 
  \label{ex:gall}%
For 
\begin{alignat*}{5} \vv_0 = (1,0) \csep \vv_1 &= (0.81,-0.27)  \end{alignat*} 
we have 
\begin{alignat*}{5}
\gisen &= 2,   \csep& \vv_2 &\approx  (0.763891087, -0.300401481), \\
\gisen &= 5/3, \csep& \vv_2 &\approx  (0.488493339, -0.371888491), \\
\gisen &= 7/5, \csep& \vv_2 &\approx  (0.283481624, -0.324994519), \\
\gisen &= 4/3, \csep& \vv_2 &\approx  (0.238714073, -0.301917762), \\ 
\gisen &= 1.01, \csep& \vv_2 &\approx (0.072809864, -0.150014313). 
\end{alignat*}
\end{example}
Generally, by choosing suitable $\vv_1,\vv_2$ it is possible to achieve $\gisen$ in the range from $-1$ to $3$. 
As $\vv_2$ passes along the circle from $0$ to $\vv_1$, $\gisen$ ranges from $1$ to some other limit that depends on $\vv_0,\vv_1$.
Pure triple shock MR not only exists, but appears to be ``generic'': 
small perturbations to $\vv_0,\vv_1$ generically allow finding a matching $\vv_2$ for the same $\gisen$. 

$\gisen$ outside the interval $\boi{-1}{3}$ do not seem attainable; such $\gisen$ are secondary from a gas dynamics perspective but relevant for nonlinear Schr\"odinger
for realizing a wider range of exponents $\pV$ for the potential $V(a)=a^{\beta-1}$.

It should be emphasized that although the triple shocks we found for potential flow may appear in MR, those Mach reflections would be unusual.
In fact since our $\vv_2$ is on the circle through $\vv_0,\vv_1,0$ it could \emph{never} be parallel to $\vv_0$, 
so such Mach reflections are unlikely to occur in ``continuous'' transition to regular reflection. 

Since $\vv_2$ cannot be parallel to $\vv_0$ the \defm{mechanical equilibrium criterion} 
(also ``von Neumann criterion''; see \citet[paragraph 45]{neumann-1943}, \cite{henderson-lozzi}) 
for transition between regular and Mach reflection does not have an obvious analogue for potential flow.
This was already noted in \citep{elling-sonic-potf,elling-detachment,elling-rrefl-lax}, 
but implicitly based on the flawed premise that potential flow does not permit pure triple shocks at all. 

In fact due to the location of $\vv_2$ on the circle through $0,\vv_1,\vv_0$ (fig.\ \ref{fig:ueq}) the reflected shock is necessarily backward, 
whereas a typical occurrence of Mach reflections is in oblique shocks reflected by a flow-parallel wall so that the reflected shock is forward, pointing downstream. 

This observation also applies to three \emph{weak} shocks in an \emph{Euler} MR. 
Although Euler does permit contacts, the contact becomes weak at a faster rate than the shocks. 
So potential flow MR is a very close approximation; the Euler reflected shock must be backward as well. 
Similar observations have already been made in past work on concrete steady or pseudo-steady MR in cases where a forward reflected shock is required. 

Of course potential flow may exhibit MR of Guderley type or other options as well 
(see \cite{guderley-1947,guderley}, also \cite{neumann-1943} paragraph 37 and 44, and \cite{vasiliev-kraiko,hunter-brio,hunter-tesdall,skews-ashworth}).

\section{Two+two interactions}
\label{section:twoplustwo}%

In this section we return to the full Euler equations. 
Having exhausted the discussion of k+1 shock interactions the obvious remaining question 
is the sign of circulation for slip lines generated by 2+2 and higher, i.e.\ at least two ``upper'' shocks with clockwise velocity and at least two ``lower'' shocks with counterclockwise velocity meeting in the same point. 2+2 interactions (fig.\ \ref{fig:fournocontacts}) are very common, e.g.\ in supersonic jet engines (\cite[fig.\ 7]{chongpeichen-tianyungao-jianhanliang-pof2019}, \cite[fig.\ 5(a1)]{kumar-aileni-rathakrishnan-pof2019}).
We show that hardly anything general can be said:
the contact produced can have clockwise, counterclockwise or zero circulation already for 2+2 interactions with polytropic $\gisen>1$ equation of state. 
In that case some small positive results can be asserted: if there is \emph{no} contact, then 2+2 is either symmetric across the flow axis if the shocks are sufficiently weak;
anti-symmetric 2+2 (meaning equal pressure ratios across \emph{antipodal} shocks) is also possible if strong shocks are permitted.

\subsection{Linearization around supersonic background}
\label{section:linsuper}

\begin{figure}
    \input{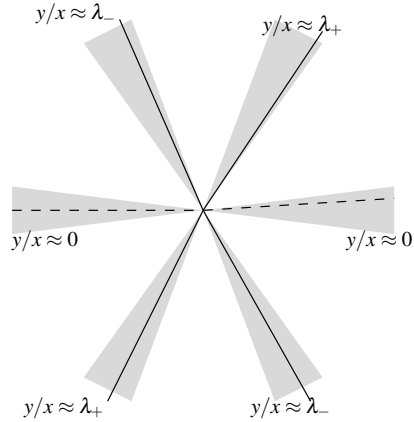}
    \caption{Weak shocks and contacts on a constant supersonic background}
    \label{fig:sectors}
\end{figure}
\begin{figure}
    \input{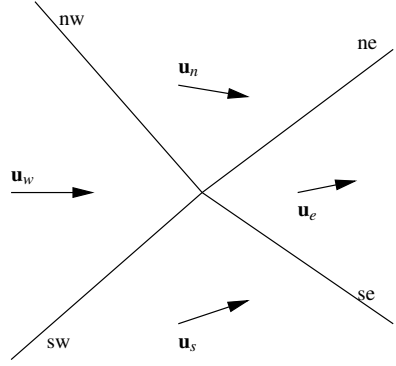}
    \caption{Four shocks without contacts}%
    \label{fig:fournocontacts}
\end{figure}
We first examine how much can be understood by merely linearizing the 2d full compressible Euler around a constant supersonic background $W=(\dens,\vx,\vy,\spm)$.
Let $\tilde W$ be the variation, then 
\begin{alignat*}{5} 
0 &= 
\subeq{\begin{bmatrix}
  \vx & \dens & 0 & 0 \\
  \csnd^2/\dens & u^x & 0 & \dati\pp\spm\dens/\dens \\
  0 & 0 & \vx & 0 \\
  0 & 0 & 0 & \vx
\end{bmatrix}}{\eqdef A^x(W)}
\begin{bmatrix}
  \tilde\dens \\
  \tilde\ve^x \\
  \tilde\ve^y \\
  \tilde\spm
\end{bmatrix}_x 
\\&+
\subeq{\begin{bmatrix}
  \vy & 0 & \dens & 0 \\
  0 & \vy & 0 & 0 \\
  \csnd^2/\dens & 0 & \vy & \dati\pp\spm\dens/\dens \\
  0 & 0 & 0 & \vy
\end{bmatrix}}{\eqdef A^y(W)}
\begin{bmatrix}
  \tilde\dens \\
  \tilde\ve^x \\
  \tilde\ve^y \\
  \tilde\spm
\end{bmatrix}_y.
\end{alignat*} 
As earlier we use rotational symmetry to take $\vy=0$ and $\vx>0$. If $\Mach=|\vv|/\csnd>1$, then $A^x$ is invertible and the eigenvalues of $(A^x)^{-1}A^y$ are 
\begin{alignat*}{5} \lambda_\pm = \pm\frac1{\sqrt{\Mach^2-1}} \csep \lambda_\spm=\lambda_\vort= 0 \end{alignat*} 
with right eigenvectors 
\begin{alignat*}{5} \rev_+ &= (-\Mach\dens,\csnd,-\csnd\sqrt{\Mach^2-1},0) ,\quad \rev_\spm = (-\dati\pp\spm\dens,0,0,\csnd^2) \\
\rev_- &= (\Mach\dens,-\csnd,-\csnd\sqrt{\Mach^2-1},0) ,\quad \rev_\vort = (0,1,0,0) \end{alignat*} 
which correspond in the nonlinear case to shocks/fans for $\pm$, shear/entropy contacts for $\vort,s$ 
(a more detailed analysis can be found in \citep{roberts-stss,elling-roberts-ii}).
$\lambda$ corresponds to the $y/x$ range at which the wave is located (fig.\ \ref{fig:sectors}.)

Let $\acr_\gtrless^*$ with $*=+,-,\spm,\vort$ be the strength of the $x\gtrless 0$ side $*$-wave (as a multiple of $\rev_*$). 
Clockwise summation yields $\sum\acr\rev=0$; in components with $\acr^i=\acr^i_<+\acr^i_>$:
\begin{alignat}{5} \Mach\dens(\acr^--\acr^+)-\acr^\spm &= 0 \\
\csnd(\acr^+-\acr^-)+\acr^\vort &= 0 \\
-\csnd\sqrt{\Mach^2-1}(\acr^++\acr^-) &= 0 \\
\csnd^2\acr^\spm &= 0
\end{alignat} 
The last equation yields $\acr^s=0$, then the first one
$\acr^+-\acr^-=0$, so the second one $\acr^\omega=0$ and the third one $\acr^-=\acr^+=0$ (we assumed $\Mach>1$). 
Hence $\acr^i_>=-\acr^i_<$ for every $i$, meaning all waves have (to leading order) the same strength as their same-type antipode.

But the last result, which is our main interest, is \emph{misleading}: we find below for the non-linearized problem 
that (with $\gamma$-law and sufficiently weak shocks) the antipodal wave strengths \emph{cannot} match; 
instead the \emph{mirror} images (across the flow axis) match. 
Besides there can be outgoing contacts even without incoming ones, as we show in section \ref{section:twoplustwocirc}; 
the strength is cubic in the shock strength, which is why we cannot observe the effect at the linear level.

\subsection{Symmetry/antisymmetric for polytropic two+two interactions}
\label{section:symm-anti}%

In special cases, namely for compressive 2+2 \emph{without} contacts (fig.\ \ref{fig:fournocontacts}) with \emph{polytropic} equation of state, 
we can say more by using non-linear arguments.
We recall the \defm{Hugoniot relation}
\begin{alignat}{5} 0 = \epm-\epm_0 + (\idens-\idens_0)\frac{\pp+\pp_0}{2} \notag\end{alignat} 
which holds for any equation of state, but for a polytropic equation of state with ratio of heats $\gisen$ we may use $e=pv/(\gisen-1)$ 
to obtain
\begin{alignat}{5} \frac{\idens}{\idens_0} = \frac{\mu^2\pp + \pp_0}{\pp+\mu^2\pp_0}, \label{eq:vv0pp0}\end{alignat} 
where $\mu^2=(\gisen-1)/(\gisen+1)$ is between $0$ and $1$ if we assume $1<\gisen<\infty$. 
Take a $\log$ on both sides and rearrange to get (with $[]$ the jump from the $0$ side)
\begin{alignat}{5} [\log\idens] = f([\log\pp]) \csep f(x) = \log\frac{\mu^2\exp x + 1}{\exp x+\mu^2}. \notag\end{alignat} 
$f$ is strictly convex:
\begin{alignat}{5} f'(x) &= \frac{\mu^2\exp x}{\mu^2\exp x+1} - \frac{\exp x}{\exp x+\mu^2} , \notag \end{alignat} 
\begin{alignat}{5}
f''(x) 
&= 
\frac{\mu^2\exp x}{(\mu^2\exp x+1)^2} 
- \frac{\exp x\mu^2}{(\exp x+\mu^2)^2} \notag
\end{alignat}
and the first denominator is smaller for $0\leq x\leq 1$ since $0<\mu^2<1$. 

Now, in absence of contacts necessarily 
\begin{alignat}{5} &
  f(\log \pp_e-\log \pp_n)+f(\log \pp_n-\log \pp_w) 
  \\&= \log \idens_e \subeq{-\log \idens_n+\log \idens_n}{=0=-\log \idens_s+\log \idens_s} - \log \idens_w \notag
  \\&= f(\log \pp_e-\log \pp_s)+f(\log \pp_s-\log \pp_w). \label{eq:ffff}\end{alignat} 
This can be written compactly as
\begin{alignat}{5} g(\lambda_n) = g(\lambda_s) \label{eq:gns}\end{alignat} 
where 
\begin{alignat}{5} 
  \lambda_n= \frac{ \log \pp_e-\log \pp_n }{ \log \pp_e-\log \pp_w } \in \boi01
\notag\end{alignat} 
and same for s in place of n, and where
\begin{alignat}{5} g(\lambda) 
&= f(\lambda(\log \pp_e-\log \pp_w)) + f((1-\lambda)(\log \pp_e-\log \pp_w)) .
\notag\end{alignat} 
$g$ is \emph{symmetric} in $\lambda$ across $\half$, 
and clearly $g$ is still strictly convex (note $\pp_e>\pp_w$ by compressiveness), 
so \eqref{eq:gns} shows
that either $\lambda_n=\lambda_s$ (symmetric case) or $\lambda_n=1-\lambda_s$ (antisymmetric case). Hence:
\begin{theorem}
  \label{th:sy-antisy}%
  For polytropic equation of state with $\gisen>1$, a compressive 2+2 shock reflection
  is either mirror-symmetric across the flow axis $\vv_w$, in particular $\pp_n=\pp_s$, 
  or anti-symmetric in the sense that pressure ratios across antipodal shocks are equal:
  \begin{alignat}{5} 
  \frac{\pp_e}{\pp_n} &= \frac{\pp_s}{\pp_w} 
  \csep
  \frac{\pp_e}{\pp_s} = \frac{\pp_n}{\pp_w} 
  \label{eq:antisy}
  \end{alignat} 
  and the same relations with $\pp$ replaced by $\idens$ (or $\epm,\spm,\hpm,\csnd$). 
\end{theorem}
(Antisymmetry does not generally imply that antipodal shock angles differ by $180^\circ$
or that pressure \defm{differences} $\pp_n-\pp_w$ and $\pp_s-\pp_e$ are equal.) 

The symmetric case is of course possible; any regular reflection at a straight slip boundary once reflected across it yields an example. 
We show below antisymmetry can also occur without producing contacts.

\subsection{Ruling out antisymmetry}

Here we try to use a velocity angle argument to rule out unsymmetric interactions. 
The angle by which shocks turn velocity is described by the shock polar, which for polytropic equation of state
has a particularly elegant form \cite[section 121]{courant-friedrichs}: 
with upstream Mach number $\Mach_0$, upstream velocity $(\vx_0,0)$ and downstream
velocity $\vv$, set $(\xi,\eta)=\vv/\vx_0$, then for admissible shocks
\begin{alignat}{5} \eta = \pm|1-\xi|\sqrt{\frac{\xi-\xi_n}{\xi_m-\xi}} \label{eq:xietapolar}\end{alignat} 
where (again $\mu^2=(\gisen-1)/(\gisen+1)$)
\begin{alignat}{5} \xi_n = \mu^2 + (1-\mu^2)\Mach_0^{-2} \notag\end{alignat} 
is $\xi$ for a normal shock while
\begin{alignat}{5} \xi_m = 1 + (1-\mu^2)\Mach_0^{-2} = \xi_n + 1-\mu^2 \notag\end{alignat} 
is the upper limit of $\xi$ on the non-compressive branch of the polar. We only consider compressive shocks with $\gisen>1$, so $\xi_n\leq\xi<1<\xi_m$ so that the absolute value can be omitted from \eqref{eq:xietapolar}. 

We already know that the 2+2 pattern is either symmetric or antisymmetric; we only have to disprove the latter. So assume it, namely that antipodal shocks have equal
compression ratio. 

As derived in section \ref{section:denscircle} the variety of $(\xi,\eta)$ achieving a fixed compression ratio $\drat=\dens/\dens_0>1$ is a circle 
centered on the horizontal axis and intersecting it in $(\xi,\eta)=(1,0)$ and $(\xi,\eta)=(1/\drat,0)$; the circle equation is 
\begin{alignat}{5} 
  \eta^2=(1-\xi)(\xi-1/\drat). \label{eq:xietacircle}
\end{alignat}
When combining \eqref{eq:xietapolar} squared and \eqref{eq:xietacircle} the factor $1-\xi$ drops out immediately ($\xi=1$ is a vanishing shock which we ignore); 
we obtain
\begin{alignat}{5} \frac{(1-\xi)(\xi-\xi_n)}{\xi_m-\xi} = \xi-1/\drat , \end{alignat} 
with solution
\begin{alignat}{5} \xi = \frac{\xi_m/\drat-\xi_n}{1/\drat+\xi_m-\xi_n-1} 
= 1 - \frac{(1-\mu^2)(1-\drat^{-1})}{\drat^{-1}-\mu^2}\Mach_0^{-2}.
\label{eq:xiMsol} \end{alignat} 
The last fraction is positive for $\drat$ between $1$ and $\mu^{-2}=(\gisen+1)/(\gisen-1)>1$, the latter being the well-known upper limit for the compression ratio across shocks.
The turning angle is $\arctan(\eta/\xi)$; \eqref{eq:xietacircle} shows
\[ (\frac{\eta}{\xi})^2 = \frac{(1-\xi)(\xi-\drat^{-1})}{\xi^2} = \frac1\drat (\frac1\xi-1)(\drat-\frac1\xi). \] 
The quadratic-in-$\xi^{-1}$ right-hand side is clearly positive between zeros at $\xi=1,\drat^{-1}$ and attains its maximum at their average
\[ \xi^{-1} = \half(1+\drat). \]
From \eqref{eq:xiMsol} we see this $\xi$ is realized for 
\[ \Mach_0^2 
=
\frac{(1+\drat^{-1})(1-\mu^2)}{\drat^{-1}-\mu^2}.
\]
By symmetry the results also hold for $\eta<0$:
\begin{proposition}
  For polytropic $\gisen>1$ equation of state the (absolute value of the) turning angle of a shock with fixed compression ratio $\drat>1$ is increasing in the upstream Mach number up to some $\Mach_{0*}$, but decreasing beyond.
  $\Mach_{0*}$ increases monotonically from $\sqrt2$ to $+\infty$ as $\drat$ increases from $1$ to $(\gisen+1)/(\gisen-1)>1$. 
\end{proposition}
Now we can complete the argument. 
For sufficiently weak shocks our linearization analysis shows (see fig.\ \ref{fig:fournocontacts}) that $\vv_s$ points right and slightly upward, $\vv_n$ right and slightly downward.
If $1<\Mach_w\leq\sqrt 2$, then 
the clockwise angle from $\vv_s$ to $\vv_w$ is larger 
than that from $\vv_e$ to $\vv_n$ since the compression ratios for the antipodal sw and ne shocks are by assumption equal 
but the upstream Mach number $\Mach_w$ of sw is by compressiveness strictly larger than the upstream Mach number $\Mach_s$ of ne.
Similarly the clockwise angle from $\vv_w$ to $\vv_n$ is larger than that from $\vv_s$ to $\vv_e$. 
Hence summing on the w side yields a larger clockwise angle from $\vv_s$ to $\vv_n$ than on the e side, contradiction! 

For $\Mach_w>\sqrt 2$ this argument (with all ``larger'' replaced by ``smaller'') still works if the shocks are chosen so weak that $\Mach_w$ as well as $\Mach_n,\Mach_s$ 
(which approach $\Mach_w$ from below as the shocks become weak) are above the respective $\Mach_{0*}$ (which decrease to $\sqrt2<\Mach_w$ as the shocks become weak). 
\begin{theorem}
  \label{th:weak-polytropic-symmetric}%
  For polytropic equation of state with $\gisen>1$, 
  a 2+2 shock reflection is necessarily mirror-symmetric across the flow axis $\vv_w$
  if all shocks are sufficiently weak (depending on $\Mach_w$).
\end{theorem}
Although the result is technically true near $\Mach_w=\sqrt2$, 
the details of our argument already suggest that anti-symmetric examples with faint shocks can be constructed near that Mach number:
\begin{example}
  A few exact values:
  \begin{alignat}{5} 
    \gisen &= \frac53 \csep \dens_w = 1 \csep \csnd_w = 1 , \notag\\
    \vv_w &= (\frac32,0) \csep \vv_n = (\frac{73}{50},-\frac{\sqrt{39}}{150}) \csep \vv_s = (\frac{49}{34},\frac{\sqrt{1463}}{646}), \notag\\
    \vv_e &= (\frac{3577}{2560}+\frac{\sqrt{1463}\sqrt{39}}{145920},-\frac{49\sqrt{39}}{7680}+\frac{73\sqrt{1463}}{48640}). \notag
  \end{alignat} 
  From these the shock normals are easily calculated by $\vn=[\vv]/|[\vv]|$, densities via $[\dens\vv\dotp\vn]=0$, 
  enthalpies by $[\hpm+\half|\vv|^2]=0$ and $\hpm=\csnd^2/(\gisen-1)$ etc.
  The density ratios are $\dens_n/\dens_w=18/17=\dens_e/\dens_s$ and $\dens_s/\dens_w=38/35$, equal for antipodal shocks as proven necessary above, but not symmetric.
  Some decimal values:\\
  \begin{tabular}{|l|l|l|l|l|}
    & w & n & s & e \\
    $\vx$ & 1.5 & 1.46 & 1.441176471 & 1.398902591 \\
    $\vy$ & 0 & -0.04163331999 & 0.05920926161 & 0.01756084113 \\
    $\csnd$ & 1 & 1.019258990 & 1.027862611 & 1.047658207 \\
    $\Mach$ & 1.5 & 1.432995442 & 1.403292831 & 1.335371403 \\ 
    $\dens$ & 1 & 1.058823529 & 1.085714286 & 1.149579832 \\
    $\pp$ & 0.6 & 0.66 & 0.6882352941 & 0.7570588236 
  \end{tabular}
\end{example}
The shock strengths can be made arbitrarily small at $\Mach_w=\sqrt2=1.414...$.

\subsection{No symmetry for ideal non-polytropic two+two}

In section \ref{section:symm-anti} we obtained that $[\log\idens]$ is a convex function of $[\log\pp]$, a property that essentially relies on a second derivative
of $\pp$ with respect to $\idens$ \emph{and} on eliminating $\spm$ from the relation using the polytropic assumption. 
For nonpolytropic pressure laws we cannot generally expect the same results. 
\eqref{eq:vv0pp0} has $\idens/\idens_0$ a function of $\pp/\pp_0$ alone, without dependence on entropy or other additional quantities,
a property not shared by most other equations of state. 
If we permit such a dependence, an analogue of \eqref{eq:gns} still constrains pressure ratios, but in a much more implicit way;
besides the derivation also requires that $g$ is strictly convex in $\lambda$, which requires strong conditions on higher derivatives that are hard to justify from thermodynamic principles.

Indeed examples of contact-free $2+2$ interactions that are neither symmetric nor antisymmetric can be found. 
We impose the ideal gas law $\pp=\dens\Temp$ (with gas constant scaled to $1$) with 
temperature $\Temp$ and energy per mass $\epm=\fod\Temp/2$.
The usual (but not universal) observed behaviour is that $\fod$ (``degrees of freedom'' per particle) increases with $\Temp$. 
We give an explicit example to demonstrate 2+2 reflections with such equations of state can be unsymmetric and not anti-symmetric:\\
\begin{example}
Some exact values: 
\begin{alignat*}{5} 
  \dens_w &= 1 \csep \pp_w = \frac{27}{100} \csep  \\ 
  \vv_w &= (1,0) \csep \vv_n = (\frac{79}{100},-\frac15) \csep \vv_e = (\frac12,-\frac{11}{1000}) \\
  \vv_s &= (\frac{2526866704111}{3099987560825},\frac{69706404456911}{371998507299000}) 
\end{alignat*} 
which yield $\vn=[\vv]/|[\vv]|$ and then $\dens_w$ yields the remaining densities via $[\dens\vv.\vn]=0$,
$\pp_w$ yields the remaining pressures via $[\dens(\vv.\vn)^2+\pp]=0$, $\pp=\dens\Temp$ yields temperatures, 
$\hpm_w$ with $[\hpm+|\vv|^2/2]$ yields enthalpies per mass, then $\epm=\hpm-\pp\idens$ yields energies per mass,
$\epm=\frac{\fod}{2}\Temp$ yields all $\fod$.
We may calculate $\csnd^2=\dati\pp\dens\epm+\dens^{-2}\pp\dati\pp\epm\dens=\Temp(1+\Temp_\epm)=\Temp(1+2/\fod)$
\emph{if} we assume the equation of state has $\fod$ constant near the chosen four $\Temp$ (other choices are possible but pointless). 
Some decimal values:\\
\begin{tabular}{|l|l|l|l|l|}
  & w & n & s & e \\
  $\vx$ & 1 & 0.79 & 0.815121562 & 0.5 \\
  $\vy$ & 0 & -0.2 & 0.187383560 & -0.011 \\
  $\dens$ & 1 & 1.667990469 & 1.599490885 & 3.026853977 \\
  $\pp$ & 0.27 & 0.48 & 0.454878437 & 0.925186655 \\
  $\Temp$ & 0.27 & 0.287771428 & 0.284389515 & 0.305659494 \\
  $\fod$ & 7.407407408 & 7.993695402 & 7.987935856 & 8.763215480 \\
  $\csnd$ & 0.585576639 & 0.599809158 & 0.596317258 & 0.612714583 \\
  $\Mach$ & 1.707718399 & 1.358637656 & 1.402579805 & 0.816238097 
\end{tabular}\\
Indeed the $\fod$ are increasing with $\Temp$. 
The four pressure ratios are distinct. 
\end{example}

\subsection{Arbitrary circulation sign for polytropic two+two}
\label{section:twoplustwocirc}

Given our results on 2+2 interactions with zero contacts, 
it is now easy to argue that 2+2 can produce outgoing contacts with nonzero vorticity even if there are no incoming contacts (left dashed line in fig.\ \ref{fig:sectors}),
even for weak shocks with $\gisen>1$ equation of state:
pick two incoming shocks of sufficiently low strength, then as discussed in section \ref{section:linsuper} 
solve the Riemann problem for steady supersonic Euler to find the outgoing waves.
Theorem \ref{th:weak-polytropic-symmetric} applies, showing that contacts can be absent 
only if the pattern is symmetric. But that is not possible if we chose the incoming shocks to have unequal strength, so there must be a nonzero outgoing contact.
Its circulation is necessarily nonzero (by section \ref{section:thermo} using $\gisen>1$ equation of state); by mirror reflection either sign is possible. 

Obviously no control over circulation can be expected for higher interactions like 3+2, 3+3, 4+2 etc.; only the k+1 case allows general results as in section \ref{section:zerocirc} and \ref{section:circulation}.

\section{Appendix}

For Euler flow (fig.\ \ref{fig:kshocks}):
\begin{proposition}
  \propmanya%
\end{proposition}
\begin{proof}
  Note: $[\vv]\dotp\vt=0$ and $\vv\dotp\vn>0$ on each side means downstream and upstream $\vv$ of a shock enclose an angle less than $90^\circ$.
  
  (a) The upstream velocity $\vv_0=(1,0)$ of shock $0$ (Mach stem) is exactly rightward, so the downstream velocity $\vv_c$ must be rightward.

  (b) By definition the Mach stem is the first shock counterclockwise from the negative horizontal axis, 
  and its upstream velocity $\vv_0$ is exactly rightward, and $0<\vv_0\dotp\vn_0=\vn^x_0$ at the Mach stem, so the Mach stem is necessarily in the open lower halfplane.

  (c) Let $i\leq k$ maximal so that shocks $1,...,i$ are in the upper left quadrant. $\vv_0=(1,0)$, so shock $1$ is backward, and since it is also compressive the downstream velocity $\vv_1$ is in the interior of the angle clockwise from $\vv_0$ to the inward tangent, in particular inside the open lower right quadrant. Then shock $2$, which is also in the open upper left quadrant, must be backward as well, with downstream velocity $\vv_2$ in the same open lower right quadrant, etc.
  Since all these velocities are in the open lower right quadrant (or on the positive horizontal axis), the radial velocity is inward in the open upper quadrant,
  therefore clockwise angular velocity is strictly increasing when moving clockwise between shocks, and since it is zero at the negative horizontal axis 
  and clockwise at each side of each shock it must be clockwise everywhere in the open upper left quadrant. 
  
  The arguments for (d)--(e) essentially use that $360^\circ$ degrees is not enough room for ``pathological'' patterns (see also \cite{roberts-stss} for similar arguments). 

  (d) Case 1: the Mach stem is in the open lower left quadrant or on the negative vertical axis. Then since $\vv_0$ is exactly rightward and the Mach stem compressive, its downstream velocity $\vv_c$ is in the angle clockwise from Mach stem \emph{antipode} to positive horizontal axis. $\vv_c$ is parallel to the contact which can (by definition) not be in the angle clockwise from negative horizontal axis to Mach stem, so the contact cannot be antipodal to $\vv_c$, it must have the same direction. In particular both are in the upper right quadrant or on the positive horizontal axis. 

  $\vv_k$ cannot be antipodal to $\vv_c$ because then it would be in the open lower left quadrant or on the negative horizontal axis, 
  hence point into shock $k$ (which is in the angle clockwise from negative horizontal axis to contact), not away from it like a downstream velocity should.
  
  Case 2: the Mach stem is in the open lower right quadrant. 
  Then $\vv_c$ is in the interior of the angle clockwise from positive horizontal axis to Mach stem, in particular in the open lower right quadrant. 
  If the contact is antipodal to it then it is in the open upper left quadrant, 
  but then by definition all upper shocks are in the angle clockwise to it from the negative horizontal axis, in particular in the upper left quadrant, 
  so (c) shows the angular velocity is not zero anywhere there, but it has to be at a contact --- contradiction. So the contact and $\vv_c$ point in the exact same direction. 

  If $\vv_k$ is antipodal to $\vv_c$, hence in the open upper left quadrant, then since $\vv_k$ is the downstream velocity of shock $k$, hence pointing away from it,
  shock $k$ and also $1,...,k-1$ are in the upper left halfplane, 
  but then (c) says $\vv_k$ must be in the open lower right quadrant, contradiction.
  
  (e) Assume the angle between upper shocks $i,i+1$ is $>180^\circ$. 
  We move in a clockwise direction from the negative horizontal axis. 
    On the downstream side of shock $i$ its downstream velocity $\vv_i$ has clockwise angular part. Moving further $\vv_i$ will flip to counterclockwise before $180^\circ$ are traversed, in particular before shock $i+1$. This flip cannot happen before the positive vertical axis since if shock $i$ is in the open upper left quadrant $\vv_i$ would be in the open lower right quadrant by (c), hence clockwise from shock $i$ to axis.

 Since $\vv_i$ cannot be counterclockwise on the upstream side of shock $i+1$, another $180^\circ$ must be traversed to let it flip back to clockwise. Shock $i+1$ would occur more than $270^\circ$ clockwise from the negative horizontal axis, but the contact must occur even later, contradicting (a)+(d).
\end{proof}

\section*{Acknowledgement}

This research was partially supported by Taiwan MOST grant 105-2115-M-001-007-MY3.

\bibliographystyle{abbrvnat}

\end{document}